\documentclass[letterpaper,11pt]{article}

\usepackage{fullpage}
\usepackage[T1]{fontenc}
\usepackage{ae,aecompl}
\usepackage{graphicx}
\usepackage{amssymb}
\usepackage{amsmath}

\usepackage{wrapfig}
\usepackage{subcaption}
\usepackage[round,authoryear]{natbib}
\usepackage{epsf}
\usepackage{color}
\usepackage{array}
\usepackage{authblk}
\usepackage{colortbl}
\definecolor{webgreen}{rgb}{0,0.4,0}
\definecolor{webbrown}{rgb}{0.6,0,0}
\definecolor{purple}{rgb}{0.5,0,0.25}
\definecolor{darkblue}{rgb}{0,0,0.7}
\definecolor{darkred}{rgb}{0.7,0,0}
\usepackage[pdfborder=false]{hyperref}
\hypersetup{colorlinks,citecolor=darkblue,filecolor=black,linkcolor=darkred,urlcolor=webgreen,pdfpagemode=None,
pdfstartview=FitH}

\newcommand{\ignore}[1]{}

\usepackage{etoolbox}
\ifundef{\abstract}{}{\patchcmd{\abstract}%
    {\quotation}{\quotation\noindent\ignorespaces}{}{}}

\newtheorem{lemma}{{\sc Lemma}}

\newtheorem{corollary}{{\sc Corollary}}
\newtheorem{theorem}{{\sc Theorem}}
\newtheorem{definition}{{\sc Definition}}

\newtheorem{example}{{\sc Example}}

\newtheorem{fact}{{\sc Fact}}

\usepackage{cleveref}
\crefname{claim}{claim}{claims}
\crefname{fact}{fact}{facts}
\crefname{algorithm}{algorithm}{algorithms}
\crefname{observation}{observation}{observations}
\crefname{equation}{equation}{equations}

\newenvironment{proof}{\noindent {\em Proof\/}:\enspace}
{\hfill $\blacksquare{}$ \medskip \\}

\DeclareMathOperator*{\argmax}{\arg\!\max}

\usepackage{pgf,tikz}
\usepackage{varwidth}
\usetikzlibrary{shapes,arrows, trees}
\usepackage{verbatim}
\usepackage{enumerate}
\usepackage{amssymb}
\usepackage{mathrsfs}
\usepackage{algorithm}
\usepackage[noend]{algorithmic}

\usepackage{multirow,bigdelim}
\usepackage{arydshln}
\usepackage{resizegather}

\title{\bf Game Theoretic Analysis of Production-Management Effort Distribution in Organizational Networks
\thanks{A preliminary version of this work has appeared as an abstract in the conference on Autonomous Agents and Multi-Agent Systems (AAMAS), 2014. The authors would like to thank Y. Narahari, David C. Parkes, Arunava Sen, and Panos Toulis for useful discussions. This work is supported by a Tata Consultancy Services doctoral fellowship.}
}

\author[1]{Swaprava Nath}
\author[2]{Balakrishnan (Murali) Narayanaswamy}

\affil[1]{\small Carnegie Mellon University, email: {\tt swapravn@cs.cmu.edu} (corresponding author)}
\affil[2]{\small Amazon.com, Inc., email: {\tt muralibn@amazon.com}}

\date{}

\begin{document}
\maketitle





\begin{abstract}
Organizations consist of individuals connected by their responsibilities, incentives, and reporting structure. These connections are aptly represented by a network, hierarchical or other, which is often used to divide tasks. A primary goal of the organization as a whole is to maximize the net productive output. Individuals in these networks trade off between their productive and managing efforts to perform these tasks and the trade-off is influenced by their positions and share of rewards in the network. Efforts of the agents here are substitutable, e.g., the increase in the productive effort by an individual in effect reduces the same of some other individual in the network, who now puts their efforts into management. The management effort of an agent improves the productivity of certain other agents in the network.

In this paper, we carry out a detailed game-theoretic analysis of individual's equilibrium split of efforts into multiple components when connected over a network. We provide a design recipe of the reward sharing scheme that maximizes the net productive output. Our results show that under the strategic behavior of the agents, it may not always be possible to achieve the optimal output using an idea from game theory called the {\em price of anarchy}.
\end{abstract}

\noindent {\bf Keywords:} Organizational networks; Network graphs; Game theory; Production efforts; Management efforts; Nash equilibrium; Price of anarchy.


%
%
%

\section{Introduction}
\label{sec:intro}
The organization of economic activity as a means for the efficient
co-ordination of effort is a cornerstone of economic theory. 
%
%
In networked organizations, agents are responsible for two processes: information flow and productive effort. A major objective of the organization is to maximize the net productive output of the networked system. However, in real organizations the individuals are responsible for multiple job roles and are rational and intelligent. They select their degree of effort which maximizes their payoff. Hence, to understand how organizations can boost their productive output, we need to understand how the individuals connected over a network split their efforts between these different roles. In particular, we study how agents in a specific model split their efforts between `effort to perform the task' \textit{versus}\ `investing effort in explaining tasks to others' depending on the amount of direct and indirect rewards. We model the agents having dual responsibilities of executing the task (production effort) and communicating the information (communication effort) to other agents. When an agent communicates with another, we call the former an {\em influencer} and the latter an {\em influencee}. Influencers improve the productivity of the influencees. Influencees, in turn, share a part of 
their rewards with the influencers, and this induces a {\em game} between the agents connected over the network.
In our model, {\em working on a task} brings a direct payoff but is more costly,
whereas {\em investing effort in explaining a task} can improve the
productivity of others (depending on the quality of communication in the network) and is less costly. The latter in turn generate additional indirect reward for an agent through reward sharing incentives.

We model the network as a directed graph, where the direction represents the direction of information flow or communication between nodes and the rewards are shared in the reverse direction. Of particular interest to us are directed trees, which represent a hierarchy, and are most prevalent in the structure of organizations and firms. In the first part of the paper, our analysis is focused on hierarchies, and in the second, we generalize our results to arbitrary directed graphs. 

The focus of this paper is to maximize the productive output of the organization, and in the process, understand the strategic behavioral dynamics of the influencing phenomenon and find the equilibrium efforts chosen by the human participants in an organizational network.
Within firms, organizational networks are often hierarchical and
there is a long history on the role of organizational structure
on economic efficiency going back to the work of~\citet{tichy1979social} on
social network analysis within organizations. More recently,  \citet{Radner92,ravasz2003hierarchical,Mookherjee2010} 
study the role of hierarchies; see~\citet{van1997state} for a survey of different perspectives. On the critical side, the work of \citet{cronin2015hierarchy} shows the adverse effects of hierarchy in human cooperation.

%
%
%
There is also a growing interest in crowdsourcing. Most relevant
here, is the ability to generate effective networks for solving challenging problems. Our model also captures some aspects of `diffusion-based task environments' where agents become aware of tasks through recruitment \citep{pickard-etal11MIT,watts2007viral}. For example, the winner of the 2009 DARPA Red Balloon Challenge adopted
an indirect reward scheme where the reward
associated with successful completion of subtasks was shared with
other agents in the network~\citep{pickard-etal11MIT}. At the same time modern massive online social networks and online gaming networks\footnote{http://www.eveonline.com/} require information and incentive propagation to organize activity. In this paper, we draw attention to the \textit{interaction} between various aspects of network influence, such as profit sharing \citep{gerhart1995employee}, information exchange \citep{bhatt2001knowledge}, and influence in networks. 

Motivated by the perspective that this phenomenon of
splitting effort into production and communication can be understood as a consequence of
the strategic behavior of the participants, we adopt a game theoretic model where individual
members in a networked organization decide on effort levels motivated
by their self interest. Agents are coordinated by incentives,
including both direct wages and indirect profit sharing. We construct
quantitative models of organizations, that are general enough to
capture social and economic networks, but specific enough for us to
obtain insightful results. We quantify the effects of reward sharing and
communication quality on the performance of 
work organizations in equilibrium. 
We then turn to the question of designing proper reward shares that can motivate people to maximize the social output of the system. We show that for stylized networks, under certain conditions, a proper incentive design can lead to the optimal social output. But when the condition is not satisfied, we capture the loss in optimality using the Price of Anarchy (PoA) framework. In particular, we provide the worst case bound on the sub-optimality.

\subsection{Overview of the Main Results}
\label{sec:overview-model}

For an easier exposition, in the first and major part of this work, we study hierarchies where the
network is a directed tree. Each agent decides how to split its effort between (i) production effort, which results in direct payoff for the agent and indirect reward to other agents on the path from the root to the
agent, and (ii) communication effort, which serves to improve the
productivity of his descendants on the tree (e.g., explaining the problem to others, conveying insights and the goals of the organization). A natural constraint is
imposed on the complementary tasks of production and communication, such that the more effort an agent
invests in production the less he can communicate. Investing production effort incurs a cost to an agent,
in return for some direct payoff.
%
%
%
But committing effort to communication
the can improve productivity of descendants, which in 
turn improves their output, should they decide to
invest effort in direct work, and thus give an agent a return on investment through an indirect payoff. 
%
%

Each agent decides, \textit{based on his position in the hierarchy}, how
to split his effort between production and 
communication, in order to maximize the sum of direct
payoff and indirect reward, accounting for the cost of effort. 
%
For most of our results we adopt an {\em exponential productivity} (EP)
model, where the quality of communication falls exponentially with effort spent in production
with a parameter $\beta$. The model has the useful property that a
pure-strategy Nash equilibrium always exists  (\Cref{thm:exp-NE-necessary}) even though the game is non-concave. In a concave game, the agents' payoffs are concave in their choices (production efforts), and a pure-strategy Nash equilibrium is guaranteed to exist~\citep{rosen1965existence}. 
The equilibrium effort given by our result explains how a `better communication' and `increase in the cost of management' incentivizes an agent to devote {\em more effort in production} -- and also how it is beneficial to spend {\em more effort in management} when the `reward sharing' increases.
We develop tight conditions for the uniqueness of the equilibrium (\Cref{thm:exponential-sufficiency}).
%
%
In addition, for the EP model of communication, the 
Nash equilibrium can be computed in time that is quadratic in the number of agents, despite the non-concave nature of the problem, by exploiting the hierarchical structure.

We then ask the question what effect this equilibrium effort level has on the total output of the hierarchical organization. We define the {\em social output} to be the sum of the individual outputs which are products of {\em productivity} and {\em production effort}. 
Our next result is that for {\em balanced hierarchies} and in the EP
model, there exists a threshold $\beta^\ast$
on communication quality parameter $\beta$ such that if the parameter is below the threshold (communication is `good enough')
then the {\em equilibrium social output} can be made equal to the {\em optimal social output} by choosing \textit{the optimal reward sharing scheme}. The phenomenon is captured by the fraction called {\em price of anarchy} (PoA), which is the ratio of the optimal and the equilibrium social output. If the reward share is not chosen appropriately, PoA can be large (\Cref{thm:large-poa}). For
$\beta$ above this threshold (`low quality' communication), we give
closed-form bounds on the PoA (\Cref{thm:poa}), which we show are tight in special
networks, e.g., single-level hierarchies. 
This highlights the importance of the design of reward sharing in organizations accounting for both network structure and communication process in order to achieve a higher network output.

%


In the second part, we consider general directed network graphs
and establish the existence of a pure-strategy Nash equilibrium and a
characterization for when this equilibrium is unique
(Theorems~\ref{thm:general-sufficient-one-NE} and
\ref{thm:general-linear-g-sufficient-unique-NE}). We provide a
geometric interpretation of these conditions in terms of the stability
properties of a suitably defined Jacobian matrix
(Figure~\ref{fig:rou-NE}). This connection between control-theoretic stability and uniqueness of Nash equilibrium in network games is an interesting property of our model.


For ease of reading, some proofs are deferred to the Appendix.

\section{A Hierarchical Model of Influencer and Influencee}
\label{sec:exponential-productivity}

\begin{wrapfigure}{r}{0.4\columnwidth}
\centering
 \includegraphics[width=0.2\columnwidth]{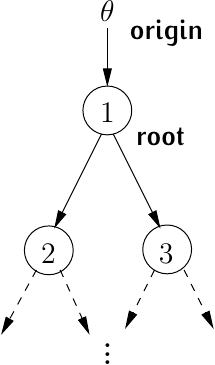}
\caption{A typical hierarchical network model.}
\label{fig:example}
\end{wrapfigure}


In this section, we formalize a specific version of the hierarchical
network model. Let $N = \{1, 2, \dots, n\}$ denote a set of agents who
are connected over a hierarchy $T$.  Each node $i$ has a set of {\em
influencers}, whose communication efforts influence his own direct payoff, and a set of {\em influencees}, whose direct payoffs are
influenced by node $i$. In turn, the production efforts of these
influencees endow agent $i$ with indirect payoffs.
The origin (denoted by node $\theta$) is a node assumed to be outside the network, and communicates perfectly with the first (root) node, denoted by $1$.

%
%
We number nodes sequentially, so that each child has a higher index than his parent, thus the adjacency matrix is an upper triangular matrix with zeros on the diagonal. Figure~\ref{fig:example} illustrates the model for an example
 hierarchical network.

The set of influencers of node $i$ consists of the nodes (excluding node $i$) on the unique path from the origin to the node, and is denoted by $P_{\theta \to i}$. The set of influencees of node $i$ consists of the nodes (again, excluding node $i$) in the subtree $T_i$ below her. 

The production effort, denoted by $x_i\in [0,1]$, of node $i$
yields a direct payoff to the node, and the particular way in which
this occurs depends on its {\em productivity}. The
remaining effort, $1-x_i$, goes to communication effort, and improves
the productivity of the influencees of the node. The constant sum of
production effort and communication effort models the constraint on an
agent's time, and therefore it is enough to write both the direct and indirect payoff of a node as a function of the production effort $x_i$. In particular, the productivity of a node, denoted by
$p_i(\mathbf{x}_{P_{\theta \to i}})$, depends on the communication
effort (and thus the production effort) of the influencers on path
$P_{\theta \to i}$ to the node. The production effort profile of these
influences is denoted by $\mathbf{x}_{P_{\theta \to i}}$.

It is useful to associate $x_i p_i(\mathbf{x}_{P_{\theta \to i}})$
with the value from the {\em direct output} of node $i$. 
The payoff to node $i$ comprises two additive terms that capture:
\smallskip

(1) the direct payoff, which depends on the value generated by 
the direct output of a node and the cost of production and communication
effort, and is modulated by the productivity of the node, and

(2) the indirect payoff, which is a fraction of the value associated
with the direct output of any influencee $j$ of the node. 
\smallskip

%

Taken together, the payoff to a single node $i$ is:
\begin{equation}
 \label{eq:general-payoff-first}
  u_i(x_i,x_{-i}) = p_i(\mathbf{x}_{P_{\theta \to i}}) f(x_i) + \sum_{j \in T_i \setminus \{i\}} h_{ij}p_j(\mathbf{x}_{P_{\theta \to j}}) x_j.
\end{equation}
The first term is the product of the direct payoff and a function $f(x_i)$ (which models production output and cost) and captures the trade-off between direct output and
cost of production and communication effort. The second term is the total indirect
payoff received by node $i$ due to the output $p_j(\mathbf{x}_{P_{\theta \to j}}) x_j$ of its influencees. 
We insist that the productivity $p_j(\cdot)$ of any node $j$ is non-decreasing
in the communication effort of each influencer, and thus non-increasing
in the production effort of each influencer, and hence we require
$\frac{\partial }{\partial
 x_i}
 p_j(\mathbf{x}_{P_{\theta \to j}})\leqslant 0$ for all nodes $j$, 
where $i$ is an influencer of $j$.

Each node $i$ receives a share $h_{ij}$ of the value of the direct
 output of influencee $j$. The model can also capture a setting where an agent can only share output he creates, i.e. the total fraction of the output an agent retains and shares with the influencers is bounded at 1. 
Let us assume
that agent $j$ retains a share $s_{jj}$ and shares $s_{ij}$ with
influencers $i
\in P_{\theta \to j}$. 
A budget-balance constraint on the amount of direct value 
that can be shared requires $\sum_{i \in P_{\theta
\to j} \cup \{j\}} s_{ij} \leqslant 1$. Assume
that $s_{jj} = \gamma>0$, for all $j$, so that each node
retains the same fraction $\gamma$ of its direct output value.
Then, the
earlier inequality can be written as, $\sum_{i \in P_{\theta \to j}}
\frac{s_{ij}}{\gamma} \leqslant \frac{1}{\gamma} - 1$. 
Define $h_{ij}:=\frac{s_{ij}}{\gamma}$. In addition to notational cleanliness, this transformation
gives the advantage of not having any upper bound on the $\sum_{i \in
P_{\theta \to j}} h_{ij}$, since any finite sum can always be
accommodated with a proper choice of $\gamma$. Let us call the matrix $H = [h_{ij}]$ containing all the reward shares as the {\em reward sharing scheme}.

To highlight our results, we focus on a specific form of the payoff model, 
namely the {\em Exponential Productivity} (EP) model. A model is an instantiation of the direct-payoff
function $f(x_i)$ and the productivity function
$p_i(\cdot)$. In particular, in the EP model:\footnote{Similar conclusions can be drawn for a reasonable choice of a concave $f$ and non-decreasing $p_i$'s. However, we pick these reasonable forms for analytical convenience and to obtain closed form expressions that enable us make clear observations and conclusions.}
\begin{align}
f(x_i)&=x_i - \frac{x_i^2}{2} - b \frac{(1-x_i)^2}{2}, \label{eq:f-in-EP}
\\
p_i(\mathbf{x}_{P_{\theta \to i}}) &= \prod_{k \in P_{\theta \to i}} \mu(C_k) e^{-\beta x_k},
\label{eq:exponential-productivity}
\end{align}
where $b \geqslant 0$ is the cost of communication, $C_k$ is 
the number of children of node $k$, function $\mu(C_k) \in [0,1]$ is assumed
to be  non-increasing, and $\beta \geqslant 0$
denotes the noise in the communication, with higher $\beta$ corresponding
to a lower quality of communication.
We assume $p_1=1$ for the root node. This models the root having perfect productivity. 
We interpret the term $\mu(C_k) e^{-\beta x_k}$ as the communication
influence of node $k$ on the agents in his subtree, and this takes values in $[0,1]$.

The direct payoff of an agent $i$ is quadratic in production effort
$x_i$, and reflects a linear benefit $x_i$ from direct production effort
but a quadratic cost $x_i^2/2$ for effort. The utility model given by \Cref{eq:general-payoff-first} resembles the utility model given by  \citet{ballester06}. However, there are a few subtle differences in our model than that in this paper: (a) the utility of agent $i$ is not concave in her production effort $x_i$ (caused by the exponential term in the productivity); thus the existence of a pure Nash equilibrium is nontrivial (for concave games pure Nash equilibrium is guaranteed to exist~\citep{rosen1965existence}), (b) each agent has two types of effort, namely production and communication, and the communication effort of an agent is {\em complementary} to the production efforts of her influencees, while the production efforts are {\em substitutable} to each other. Also, the complementarity is nonlinear. In \Cref{sec:general-model}, we address quite general nonlinear complementarity. This is a step forward to the multidimensional effort distribution with nonlinear correlation between the efforts among agents. We 
chose this particular form to capture a realistic organizational hierarchy. (c) In addition, we also consider the cost due to communication, captured by $b (1-x_i)^2/2$.

%

The productivity of node $j$, given by
$p_j(\mathbf{x}_{P_{\theta \to j}})$, where $j \in T_i \setminus
\{i\}$ warrants careful observation. Here we explain the components of this function and the reasons for choosing them. Consider $\mu(C_k)$, which is
non-increasing in the number of children. The set $C_k$ captures the idea that
the effect of the communication effort is reduced if the node has more children to communicate with.
%
%
An increase in production effort $x_k$ reduces the productivity of 
influencees of node $k$. In particular, the exponential term in the productivity captures two effects:
(a) a linear decrease in production effort gives exponential gain
in the productivity of influencee, which captures the importance of communication and management in organizations \citep{Allen2007Innovation}.
Smaller values of $\beta$ model better communication and a stronger positive effect on an influencee.
(b) We can approximate other models  by choosing $\beta$ appropriately. Linear productivity corresponds to small values of $\beta$. This property is useful when the effects of production and communication on the payoff are equally important. For large $\beta$ there is very small communication quality between agents and the value of communication effort is low.

The successive product of these exponential terms in the path from root to a node reflects the fact that a change in the production effort of an agent affects the productivity of the entire subtree below her. We note that the productivity of node $j$, where $j \in T_i \setminus
\{i\}$, is not a concave function of $x_i$, leading to the payoff function $u_i$ to be non-concave in $x_i$. Hence the existence of
a Nash equilibrium is not guaranteed a priori through known results on concave games~\citep{rosen1965existence}. 
%
In the next
section we will demonstrate the required conditions on existence and
uniqueness of a Nash equilibrium.
For brevity of notation, we will drop the arguments of
productivity $p_i$ at certain places where it is clear from the context. 




Our results on existence, uniqueness and their interpretations generalize to other network structures beyond hierarchies, which we show in the later part of the paper. However, despite the
mathematical simplicity of the EP model, it allows for obtaining interesting results on the importance of influence, both communication and
incentives, and gives insight on outcome efforts in a networked organization.

\subsection{Main Results}
\label{sec:main-results}

The effect of communication efforts between nodes $i$ and $j$, where $i \in P_{\theta \to j}$ is captured by 
the fractional productivity $\frac{p_j}{p_i}$ defined as,
$p_{ij}(\mathbf{x}_{P_{i_- \to j}}) = \prod_{k \in P_{i_- \to j}}
\mu(C_k) e^{-\beta x_k}$, (the node $i_-$ is the parent of $i$ in the
hierarchy). This term is dependent only on the production efforts in the path segment between $i$ and $j$ and accounts for `local' effects. We show in the following theorem that the Nash equilibrium production effort of node $i$ depends on this local information from all its descendants.
\begin{theorem}[Existence of Pure Nash Equilibrium]
 \label{thm:exp-NE-necessary}
 A pure Nash equilibrium always exists in the effort game in the EP model, 
and is given by the production effort profile $(x^*_i, x^*_{-i})$ that satisfies,\footnote{Define $x^+ := \max \{0,x\}$.}
 \begin{align}
   \label{eq:exponential-NE}
    x_i^* = \left [ 1 - \frac{\beta}{1+b} \sum_{j \in T_i \setminus \{i\}} h_{ij} p_{ij}(\mathbf{x}^*_{P_{i_- \to j}}) x_j^* \right ]^+.
  \end{align}
\end{theorem}
%
 The proof of this theorem uses the hierarchical structure of the network and the fact that the productivity functions ($p_i$'s) are bounded. We present the proof in Appendix~\ref{app:EP-proofs}.

This theorem shows that the EP model allows us to guarantee the
existence of (at least one) Nash equilibrium. In particular, we can make certain observations on the equilibrium
production effort, some of which are intuitive.
\begin{itemize}
\item If communication
improves, i.e., $\beta$ becomes small, the production effort of each
node increases.
\item If the cost of management $b$ increases, the production
effort of each node increases. 
\item When reward
sharing ($h_{ij}$) is large, agents reduce production effort and
focus more on communication effort, which is more productive in terms of payoffs.
\item The computation of a Nash equilibrium at any node depends only on the production efforts of the nodes in its subtree. Thus, we can employ a backward induction algorithm which exploits this property that helps in an efficient computation of the equilibrium (this will be shown formally in the corollaries later in this section).
\end{itemize}
We now turn to establishing conditions for the uniqueness
of this Nash equilibrium.
Let us define the maximum amount of reward share that any node $i$ can accumulate from a hierarchy $T$ given a reward sharing scheme $H$ as,
$h_{\max}(T) = \sup_i \sum_{j \in T_i \setminus \{i\}} h_{ij}$. 
%
We also define the {\em effort update function} as follows.
\begin{definition}[Effort Update Function (EUF)]
 \label{def:EUF}
 Let the function $F: [0,1]^n \to [0,1]^n$ be defined as,
  $$F_i(\mathbf{x}) = \left [ 1 - \frac{\beta}{1+b} \sum_{j \in T_i \setminus \{i\}} h_{ij} p_{ij}(\mathbf{x}_{P_{i_- \to j}}) x_j \right ]^+.$$
\end{definition}
Note that the RHS of the above expression contains the production efforts of all the agents in the subtree of agent $i$. This function is a prescription of the choice of the production effort of agent $i$, given a certain effort profile of the agents below $i$ in the hierarchy (\Cref{thm:exp-NE-necessary}). Hence the name `effort update'.

%
\begin{theorem}[Sufficiency for Uniqueness]
 \label{thm:exponential-sufficiency}
  If $\beta < \sqrt{\frac{1+b}{h_{\max}(T)}}$, the Nash equilibrium effort profile $(x^*_i, x^*_{-i})$ is unique and is given by Equation (\ref{eq:exponential-NE}).  
\end{theorem}
%
 The proof of this theorem shows
that $F$ is a contraction, and is given in Appendix~\ref{app:EP-proofs}.
%
\begin{theorem}[Tightness]
\label{thm:tightness}
  The sufficient condition of Theorem~\ref{thm:exponential-sufficiency} is tight.
\end{theorem}
\begin{wrapfigure}{R}{0.3\columnwidth}
\centering
 \includegraphics[width=0.15\columnwidth]{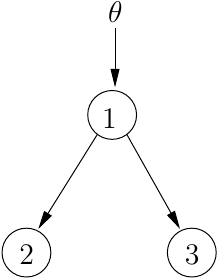}
\caption{Tightness of the sufficiency (Theorem~\ref{thm:exponential-sufficiency}).}
\label{fig:exnet}
\end{wrapfigure}
\begin{proof}
Consider a three-node hierarchy with nodes 2 and 3 being the children of node 1 (Figure~\ref{fig:exnet}). We show that if the sufficient condition 
is just violated, it results in multiple equilibria. Let $b=0$, and $h_{12} = h_{13} = 0.25$, therefore $h_{\max}(T) = 0.25$. Theorem~\ref{thm:exponential-sufficiency} requires that $\beta < 1/\sqrt{0.25} = 2$. We choose $\beta = 2$. The equilibrium efforts for node 2 and 3 are $1$. Node 1 solves the following equation to find the equilibria.
\[1-x_1 = e^{-2 x_1}.\]
This equation has multiple solutions, $x_1 = 0, 0.797$, showing non-uniqueness.
\end{proof}

The uniqueness condition indicates that the communication quality
needs to be `good enough' (small $\beta$) to ensure uniqueness of an
equilibrium. 
%
%
It is worth noting that the uniqueness condition ensures the convergence of 
the best response dynamics,
in which all the players start from any arbitrary effort profile
$\mathbf{x}_{\text{init}}$, and sequentially update their efforts via
the function $F$, to the unique equilibrium. This is a
consequence of the fact that $F$ is a contraction.


We now turn to the computational complexity of
a Nash equilibrium. If there are multiple NE, these complexity results hold for computing {\em a} NE. Recall that the equilibrium computation of an agent requires only the production efforts and the reward structure of its subtree, and we can take advantage of the backward induction. This observation leads to the following corollary.
%
\begin{corollary}
 \label{cor:node-NE-complexity}
 The worst-case complexity of computing the equilibrium effort of node $i$ is $O(|T_i|^2)$. As a result, in this hierarchical model, the worst-case complexity of computing the equilibrium efforts of the whole network is $O(n^2)$.
\end{corollary}

Using the structure of the Nash equilibrium obtained in this section, we now address the question of the total productive output generated in equilibrium.

\section{Maximizing the Productive Output of the Network}
\label{sec:poa-genl}
In our model, the equilibrium behavior of the agents are tightly coupled with the network structure and the reward sharing scheme as seen from \Cref{eq:exponential-NE}. In this section, we look at how the equilibrium behavior affects the {\em social output} of the hierarchy $T$ for a given effort vector $\mathbf{x} \in [0,1]^n$, defined as follows.
\begin{align}
 SO(\mathbf{x}, T) &= \sum_{i \in N} p_i(\mathbf{x}_{P_{\theta \to i}}) x_i
\end{align}
This quantity captures the sum of the output of each individual agents in the network, where the output of each agent is the product of their productivity and production effort.
For a given hierarchy $T$, define the optimal effort vector as $\mathbf{x}^{\text{OPT}} \in \argmax_{\mathbf{x}} SO(\mathbf{x}, T)$. 
This is the production effort profile across the network that maximizes
the total direct output value, considering also the effect of communication
effort (induced by lower production effort) on the productivity of other nodes. Ideally a planner (the management of an organization) would like to achieve this maximal social output for the given hierarchy. However, the strategic choice of the individuals might not always lead to this global performance.
The question we address in this section is how the Nash equilibrium effort level
$\mathbf{x}^*$ performs in comparison to the socially optimal outcome
$\mathbf{x}^{\text{OPT}}$.

Note that computing $\mathbf{x}^{\text{OPT}}$ is easy for the EP model. Finding the maxima of $SO(\mathbf{x}, T)$ can be done via backward induction on the levels in the hierarchy and solving nonlinear equations of single variable at each stage.

We will consider cases where the equilibrium is unique, hence, the {\em price of anarchy} \citep{koutsoupias1999worst} is given by:
\begin{align}
\text{PoA} &= \frac{SO(\mathbf{x}^{\text{OPT}}, T)}{SO(\mathbf{x}^*, T)}.
\end{align}

This quantity measures the degree of efficiency of the network. Making PoA equal to unity would be the ideal achievement for the designer. However, that may not always be possible given the parameters of the model. In such a case, we provide a design procedure of the reward sharing scheme that yields the maximum social output.

We note that the equilibrium effort profile $\mathbf{x}^*$ depends on the reward sharing scheme $H$, while $\mathbf{x}^{\text{OPT}}$ does not. The goal of this section is to understand how one can engineer the $H$ to reduce the PoA (thereby making the social output closer to the optimal). The following example
shows that if the reward sharing is not properly designed, the PoA can
be arbitrarily large. 
We first consider a single-level hierarchy (see
Figure~\ref{fig:star}). To simplify the analysis, we also
assume that the function $\mu(C_1) = 1$, irrespective of the number of
children of node 1. By symmetry, we consider a single value $h$, such that
$h_{12}= h_{13}= \ldots= h_{1n}= h$. We refer to this model as {\tt FLAT}.
We will return to
this model later as well, after presenting our results for more
general balanced hierarchies.
We first consider what happens when there is bad communication
($\beta$ large) and no profit sharing ($h=0$), between node 1 and its children.

\begin{wrapfigure}{R}{0.3\columnwidth}
\centering
 \includegraphics[width=0.2\columnwidth]{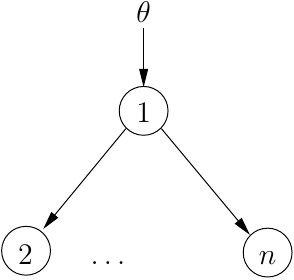}
\caption{{\tt FLAT} hierarchy.}
\label{fig:star}
\end{wrapfigure}


\begin{example}[Large PoA] \label{thm:large-poa}
For $n \geqslant 3$, the PoA is $\frac{n-1}{2}$ 
in the {\tt FLAT} hierarchy when $\beta = \ln (n-1)$
and $h=0$.
For {\tt FLAT}, the social output is given by, $SO(\mathbf{x}, \text{\tt FLAT}) = \sum_{i = 2}^{n} e^{-\beta x_1} x_i + x_1$. We see that $\beta = \ln (n-1) \geqslant -\ln \left( 1 - \frac{1}{n-1}\right)$, for all $n \geqslant 3$. The optimal effort profile $\mathbf{x}^{\text{OPT}} = (0, 1, \dots, 1)$ maximizes the social output (stated in \Cref{cor:opt-effort}, for a proof see \Cref{lem:optimal-effort} in \Cref{app:B}). Hence the optimal social output is $n-1$.
However, for reward sharing factor $h=0$, we get the equilibrium effort profile from \Cref{eq:exponential-NE} to be $\mathbf{x}^* = (1, 1, \dots, 1)$. This yields a social output of $(n-1) e^{-\beta} + 1$. Hence the PoA is $\frac{n-1}{(n-1) e^{-\ln (n-1)} + 1} = \frac{n-1}{2}$.
\end{example}

%


However, if $h$ is chosen appropriately, e.g., if it were chosen to be large positive, the equilibrium effort profile would have been closer to that of the optimal -- leading to PoA being closer to $1$.

This raises a natural question: \textit{is it always possible to
design a suitable reward sharing scheme that can make PoA $=1$ for any
given hierarchy?} In order to answer that, we define the {\em
stability} of an effort profile $\mathbf{x}$.
\begin{definition}[Stable Effort Vector]
 An effort profile $\mathbf{x} = (x_1, \dots, x_n)$ is {\em stable}, represented by $\mathbf{x} \in S$, if $\mathbf{x} \geqslant \mathbf{0}$, and there exists a reward sharing matrix $H = [h_{ij}], \ h_{ij} \geqslant 0$, such that,
\begin{equation}
  \label{eq:admissible}
  \begin{aligned}
  \sum_{j \in T_i \setminus \{i\}} a_{ij}(\mathbf{x}) h_{ij} \geqslant 1 - x_i; \quad \sum_{j \in T_i \setminus \{i\}} h_{ij} \leqslant \frac{1+b}{\beta^2}, \ \forall i \in N.
 \end{aligned}
\end{equation}
Where, $a_{ij}(\mathbf{x}) = \frac{\beta}{1+b} p_{ij}(\mathbf{x}_{P_{i_- \to j}}) x_j$, for all $j \in T_i \setminus \{i\}$, and zero otherwise. If such a solution $H^*$ exists, we call it a \emph{stable reward sharing matrix}.
\end{definition}

The inequalities capture a required balance between incentives and
information flow. In the first inequality, for a fixed communication factor $\beta$ and cost coefficient $b$, the term $a_{ij}(\cdot)$ is proportional to the
fractional output (fractional productivity $\times$ production effort) of an agent $j$.
After multiplying with $h_{ij}$, this is the effective indirect
output that $i$ receives from $j$. The RHS of the inequality can be
interpreted as the communication effort of agent $i$. Hence, this
inequality says that the total indirect benefit should be at least
equal to the effort put in by a node for communicating the information
to its subtree. 
If we consider that the agents share information based on the reward share they receive, the flow of information and reward forms a closed loop.
The second inequality says that the closed loop `gain' of the
information flow ($\beta^2$) and the reward share accumulated by agent $i$ ($\sum_{j \in
T_i \setminus \{i\}} h_{ij}$) should be bounded by the cost of sharing the information. The closed loop `gain' is essentially the reward that an agent accumulates due to his communication effort \textit{through} his descendants. We can connect a stable effort vector with the Nash equilibrium of the effort game.

\begin{lemma}[Stability-Nash Relationship]
  \label{lem:stability-nash}
 If an effort profile $\mathbf{x} = (x_1, \dots, x_n)$ is stable, it is the unique Nash equilibrium of the effort game with the corresponding stable reward sharing matrix.
\end{lemma}

\begin{proof}
 Let $\mathbf{x}$ is a stable effort profile. So, there exists a stable reward sharing matrix corresponding to it. Let $H = [h_{ij}], \ h_{ij} \geqslant 0$ be the matrix, s.t.\ the inequalities corresponding to \Cref{eq:admissible} are satisfied with $\mathbf{x}$. Since $\mathbf{x} \geqslant \mathbf{0}$, reorganizing the first inequality of \Cref{eq:admissible} and noting the fact that $x_i \geqslant 0, \ \forall i \in N$, we get,
  \[x_i = \left[ 1 - \sum_{j \in T_i \setminus \{i\}} a_{ij}(\mathbf{x}) h_{ij} \right]^+, \ \forall i \in N.\]
 Under the condition given by the second inequality of \Cref{eq:admissible}, the Nash equilibrium is unique and is given by the above expression (recall \Cref{thm:exponential-sufficiency}). Hence, $\mathbf{x}$ is the unique Nash equilibrium of this game.
\end{proof}

Now it is straightforward to see that the stability of $\mathbf{x}^{\text{OPT}}$ is sufficient for PoA to be $1$. This is because now the $H$ that makes the $\mathbf{x}^{\text{OPT}}$ vector stable can be used as the reward sharing scheme, and for that $H$ the equilibrium effort profile will coincide with $\mathbf{x}^{\text{OPT}}$. In other words, the optimal effort vector can be supported in equilibrium by a suitable reward sharing scheme. Hence, the following lemma is immediate.

\begin{lemma}[No Anarchy]
\label{lem:noAnarchy}
A stable reward sharing scheme corresponding to $\mathbf{x}^{\text{OPT}}$ yields a PoA of 1.
\end{lemma}
%
%

A couple of important questions are then: {\em how efficiently can we check if a given effort profile $\mathbf{x}$ is stable or not? And how to choose a reward sharing scheme that makes the effort profile stable?}
The answer is that we can solve the following feasibility linear program (LP) for a given effort profile:

\begin{equation}
 \label{eq:feasLP}   
 \begin{array}{cc}
  \min & 1 \\
  \mbox{ s.t. } & \left . \begin{array}{rcl}
  \sum_{j \in T_i \setminus \{i\}} a_{ij}(\mathbf{x}) h_{ij} & \geqslant & 1 - x_i, \\
  \sum_{j \in T_i \setminus \{i\}} h_{ij} & \leqslant & \frac{1+b}{\beta^2}, \\
      h_{ij} & \geqslant & 0, \forall j,
                  \end{array}
  \right \} \ \forall i \in N.
 \end{array}
\end{equation}
If a solution exists to the above LP, we conclude that $\mathbf{x}$ is stable and declare the corresponding $H$ to be the resulting reward sharing scheme.
Linear programs can be efficiently solved and therefore checking an effort profile for stability can be done efficiently.


\paragraph{A Note on the Reward Share Design}
This condition gives us a recipe of the design of the reward sharing scheme. However, the next question is: {\em what happens when the $\mathbf{x}^{\text{OPT}}$ is unstable?}
If the above feasibility LP does not return any solution matrix $H$, we conclude that $\mathbf{x}^{\text{OPT}} \notin S$, where $S$ is the set of all stable effort vectors. In such a scenario, we cannot guarantee PoA to be unity. However, for any given reward sharing matrix $H$, there is an equilibrium effort profile $\mathbf{x}^*(H)$. We can, therefore, solve for $H_{\max} \in \argmax_{H : \mathbf{x}^*(H) \in S} SO(\mathbf{x}^*(H))$ which leads to an equilibrium effort profile $\mathbf{x}^*(H_{\max})$ that lies in the stable set and maximize the social output. Therefore, when we cannot find a reward sharing scheme to achieve the optimal social output, $H_{\max}$ is our best bet. Computing $H_{\max}$ for general hierarchies may be a hard problem, and we leave that as an interesting future work. However, for certain special classes of hierarchies, it is possible to derive bounds on the PoA (thereby providing a design recipe for $H$ to achieve a lower bound on the social output). In the following section, we do the same for the balanced hierarchies. The price of anarchy analysis, therefore, serves as a means to find the optimal reward sharing scheme that gives a theoretical guarantee on the social output of the system.

\subsection{Price of Anarchy in Balanced Hierarchies}
\label{sec:poa}

In this section we consider a simple yet representative class of
hierarchies, namely the balanced hierarchies, and analyze the effect
of communication on PoA and provide efficient bounds. Hierarchies in
organizations are often (nearly) balanced, and the {\tt FLAT} or
linear networks are special cases of the balanced hierarchy (depth = 1
or degree = 1). Hence, the class of balanced hierarchies can generate
useful insights. In addition, the symmetry in balanced hierarchies
allows us to obtain interpretable closed-form bounds and understand
the relative importance of different parameters.

We consider a balanced $d$-ary tree of depth $D$. By symmetry, the
efforts of the nodes that are at the same level of the hierarchy are
same at both equilibrium and optimality. This happens because of the fact that in the EP model, both the equilibrium and optimal effort profile computation follows a backward induction method starting from the leaves towards the root. Since the nodes in the same level of the hierarchy is symmetric in the backward induction steps, they have identical effort profiles. 

With a little abuse of notation, we denote the efforts of each node at
level $i$ by $x_i$. We start numbering the levels from root, hence,
there are $D+1$ levels. Note that there are a few interesting special
cases of this model, namely (a) $d = 2$: balanced binary tree, (b) $D
= 1$: flat hierarchy, (c) $d = 1$: line. We assume, for notational simplicity only, that the function $\mu(C_k) = 1$, for all $C_k$, though our results generalize. This function is the coefficient of the productivity function. $\mu(C_k) = 1$ also models organizations where each manager is assigned a small team and there is no attenuation in productivity due to the number of children. In order to present the {\em price of anarchy} (PoA) results, we
define the set $\xi$:
\begin{equation}
 \xi(\beta) = \left\{ x : x = \left[ 1 - \frac{1}{\beta} e^{-\beta x}\right]^+ \right\}. \label{eq:xi-defn}
\end{equation}

This set is the set of possible equilibrium effort levels for agents
at the penultimate level of the EP model hierarchy when $\beta >
1$. Note that this set is a singleton, when $\beta > 1$. Depending on
$\beta$, we define a lower bound $\phi(d,\beta)$ on the contribution
of an agent toward the social output,
and a
sequence of nested functions $t_i$, where $d$ is the degree of each node.

 \begin{equation} \label{eq:t_D}
 \begin{aligned}
  \phi(d,\beta) &= \max \left \{ \frac{1}{\beta} (1 + \ln (d \beta)), d \beta +  (1-d \beta) \xi(\beta) \right \}, \\
  t_1(d,\beta) &= \phi(d,\beta), t_2(d,\beta) = \phi(d \cdot \phi(d,\beta),\beta), \dots, t_D(d,\beta) &= \phi(d \cdot t_{D-1}(d,\beta),\beta).
 \end{aligned}
 \end{equation}

\begin{theorem}[Price of Anarchy]
 \label{thm:poa}
For a balanced $d$-ary hierarchy with depth $D$, as $\beta$ increases, we can show the following price of anarchy results.
  \begin{equation}
    \begin{aligned}
      \text{ When } 0 \leqslant \beta \leqslant 1, & \qquad \text{ PoA } = 1, \\
      \text{ and when } 1 < \beta < \infty, & \qquad \text{ PoA } \leqslant \frac{d^D}{t_D(d,\beta)}.
    \end{aligned}
  \end{equation}
\end{theorem}

\begin{proof}
 The proof is constructive and sets the $H$ matrix appropriately to achieve the bounds on PoA. The $H$ matrix constructed this way acts as the reward sharing scheme to achieve a reasonable enough social output. For details, see Appendix~\ref{app:B}.
\end{proof}

\begin{wrapfigure}{R}{0.5\textwidth}
 \centering
 \includegraphics[width=0.4\columnwidth]{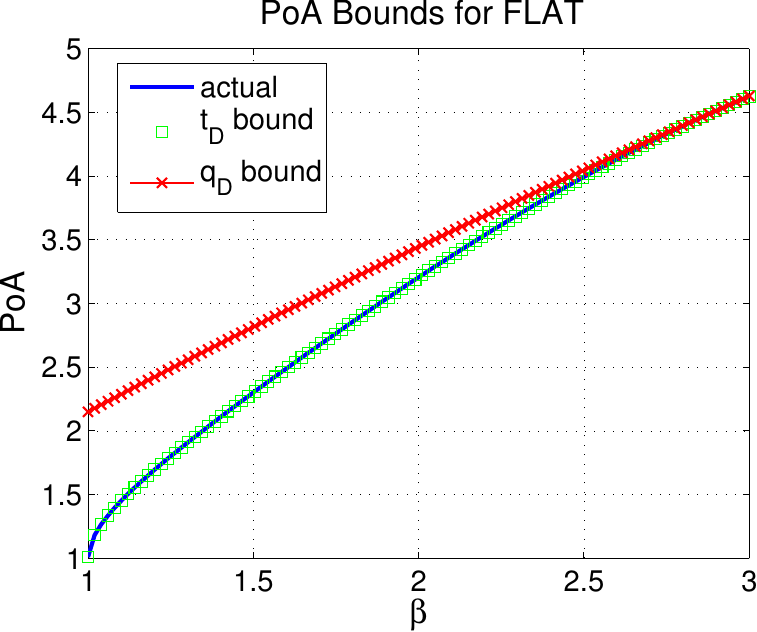}
\caption{Bounds on PoA for {\tt FLAT}, $d=6,D=1$.}
\label{fig:poaStar}
\end{wrapfigure}
As opposed to our choice of lower bound $\phi$, a na\"ive lower bound
of $\frac{1}{\beta} (1 + \ln (d \beta))$ can also be used. 
However, this gives a weaker bound on the PoA. As an example, we demonstrate the weakness of the bound for {\tt FLAT} (recall Figure~\ref{fig:star}) in
Figure~\ref{fig:poaStar} (the {\tt FLAT} hierarchy is a balanced
tree with $D=1, d=n-1$) -- where $q_D$ is same as $t_D$ when $\phi(d,\beta)$ is defined to be $\frac{1}{\beta} (1 + \ln (d \beta))$. Figure~\ref{fig:poaStar} shows that the bound
given by our analysis is tight for {\tt FLAT},
indicating the value of the analysis and also gives intuition to the shape of the effect of $\beta$ on the PoA.

We can then have the following corollaries of Theorem~\ref{thm:poa},
\begin{corollary}[Optimal Effort]
\label{cor:opt-effort}
 For the {\tt FLAT} hierarchy, if $0 \leqslant \beta < -\ln \left( 1 - \frac{1}{n} \right)$, the optimal effort profile is where all nodes put unit effort. When $-\ln \left( 1 - \frac{1}{n} \right) \leqslant \beta < \infty$, the optimal changes to the profile where the root node puts zero effort and each other node puts unit effort.
\end{corollary}

\begin{corollary}
 For the {\tt FLAT} hierarchy, when $0 \leqslant \beta \leqslant 1$, PoA $= 1$, and when $1 < \beta < \infty$, PoA $\leqslant \frac{n}{\phi(d,\beta)}$.
\end{corollary}

The second corollary above makes rigorous the intuition that when
$\beta$ is small enough the optimal $\mathbf{x}$ can be achieved in the equilibrium of a strategic play of the agents by
choosing a small enough reward share $h$. However, when $\beta$ grows,
in order to ensure uniqueness of the Nash equilibrium, the choice of
$h$ becomes limited (as it has to satisfy $\leqslant (1+b) / \beta^2$)
resulting in a PoA, as captured in Figure~\ref{fig:poaStar}.

\Cref{thm:poa} also gives a theoretical justification of the usefulness of good communication on productive output in balanced hierarchies. When communication is good, i.e., $\beta$ is small, it is possible to design reward sharing schemes to achieve optimal effort profile in equilibrium -- which ceases to be the case when communication worsens ($\beta$ becomes large).

\section{A General Network Model of Influencer and Influencee}
\label{sec:general-model}

In this section, we show that the results on existence and uniqueness
of a pure strategy Nash equilibrium generalize to a much broader
setting of agents as influencer and influencees interacting over an arbitrary network.

Suppose that the agents are connected over a (possibly non-hierarchical)
network $G$. Each node $i$ has a set of influencers, denoted by $R_i$ (generalizing $P_{\theta \to i}$),
and a set of influencees, $E_i$ (generalizing $T_i \setminus \{i\}$). We import the notation from Section~\ref{sec:exponential-productivity} with their exact or analogous meanings for productivity $p_i(\mathbf{x}_{R_i})$ and reward sharing scheme $H$. 
Now, the payoff
function of agent $i$ is given by,

\begin{equation}
 \label{eq:general-payoff}
  u_i(x_i,x_{-i}) = p_i(\mathbf{x}_{R_i}) f(x_i) + \sum_{j \in E_i} h_{ij} p_j(\mathbf{x}_{R_j}) x_j .
\end{equation}


We assume that $f$ is a strictly concave function, and is continuously
differentiable. We will refer to the product of effort $x_i$ and
productivity $p_i(\mathbf{x}_{R_i})$ as the output, and denote it by $y_i$.
In this context, we do not impose any condition on the nature of the
productivity function $p_i(\cdot)$, and as before, this game is also not
necessarily a concave game and the existence of a Nash equilibrium is
not always guaranteed. 

\subsection{Results}

The payoff function given by Equation (\ref{eq:general-payoff}) induces a game between the influencers and the influencees. In addition, as before, every agent faces a trade-off when deciding how much production and communication effort to exert. We will use the following facts which are well known from real analysis~\citep{rudin1964principles}.

\begin{fact}
 If a function is continuously differentiable and strictly concave, its derivative is continuous and monotone decreasing.
\end{fact}

\begin{fact}
 \label{fact:invertible}
 A continuous and monotone decreasing function is invertible and the inverse is also continuous and monotone decreasing.
\end{fact}

Using the above two facts, we see that the inverse of $f'$ exists and is monotone decreasing. Let us denote $f'^{-1}$ by $\ell$. Let us define two functions $g$ and $T$ similar to that defined in Section~\ref{sec:exponential-productivity}.

\begin{align}
 g_i(\mathbf{x}) &= \sum_{j \in E_i} h_{ij} \left . \left ( - \frac{1}{p_i(\mathbf{x}_{R_i})} \frac{\partial p_j(\mathbf{x}_{R_j})}{\partial x_i} x_j \right ) \right |_{\mathbf{x}}, \label{eq:g-general} \\
T(x) &= \min \{ \max \{0, x\}, 1\}. \label{eq:T-general}
\end{align}

\begin{fact}
 \label{fact:T-continuous}
 The function $T$ is continuous.
\end{fact}

\begin{lemma}[Necessary condition for Nash equilibrium]
\label{lem:general-necessary-NE}
 If a Nash equilibrium exists for the effort game in a influencer-influencee network, the effort profile $(x^*_i, x^*_{-i})$ must satisfy,
\begin{align}
 \label{eq:general-NE}
  x_i^* = T \circ \ell \circ g_i(\mathbf{x}^*), \ \forall i \in N.
\end{align}
\end{lemma}

To illustrate what this necessary condition means, let us assume, for simplicity, that we do not hit the edges of the truncation function $T$. Therefore we can rewrite Equation (\ref{eq:general-NE}) as, 

\begin{equation}
 \label{eq:general-NE-reorg}
  f'(x^*_i) = \sum_{j \in E_i} h_{ij} \left . \left ( - \frac{1}{p_i} \frac{\partial p_j}{\partial x_i} x_j \right ) \right |_{\mathbf{x}^*} = \sum_{j \in E_i} h_{ij} \left . \left ( - \frac{1}{p_i} \frac{\partial y_j}{\partial x_i} \right ) \right |_{\mathbf{x}^*}
\end{equation}

Where $y_j = p_j x_j$ is the output of node $j$. We have dropped the arguments of $p_i$ and $p_j$ for brevity of notation. The expression on the LHS is the rate of change of direct benefit for agent $i$. The RHS is the rate at which the passive output of agent $i$ changes w.r.t.\ his effort $x_i$ and productivity $p_i$. If the LHS is larger, the agent would gain more at the margin by increasing $x_i$. This is because the derivative $\left ( - \partial y_j/\partial x_i \right )$ is non-negative since $\partial p_j/\partial x_i$ is always non-positive. Similarly, if the RHS was larger, the agent could gain at the margin by decreasing $x_i$. Hence Equation (\ref{eq:general-NE-reorg}) resembles a rate balance equation (or demand-supply curve) where the rate of effective direct payoff matches the rate of passive payoffs.

\begin{wrapfigure}{R}{0.5\columnwidth}
\centering
 \includegraphics[width=0.4\columnwidth]{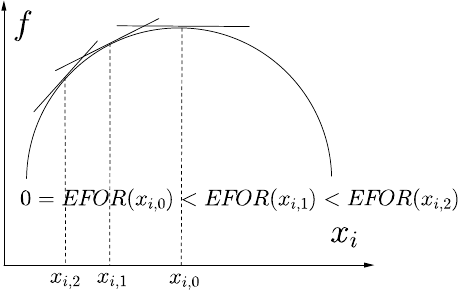}
\caption{Impact on the production effort of a node having higher EFOR from the influencees.}
\label{fig:general-NE}
\end{wrapfigure}

Let us define the {\em effective fractional output rate} (EFOR) at $\mathbf{x}$ as $\sum_{j \in E_i} h_{ij} \left . \left ( - \frac{1}{p_i} \frac{\partial y_j}{\partial x_i} \right ) \right |_{\mathbf{x}}$.
In some settings, e.g., if $p_i(\mathbf{x}_{R_i}) = \prod_{k \in R_i} (1-x_k)$, the fractional output rate $\left ( - \frac{1}{p_i} \frac{\partial y_j}{\partial x_i} \right )$ can be independent of the production effort of $i$, i.e., $x_i$. In such settings, Equation (\ref{eq:general-NE-reorg}) shows that if the EFOR of node $i$ increases, the equilibrium for node $i$ will move in a direction that decreases production effort $x_i$. This happens because the slope of $f$ in equilibrium is always non-negative and its increase leads to a smaller $x_i$ because of the concavity of $f$. This phenomenon is graphically shown in Figure \ref{fig:general-NE}. This shows that the nodes having a higher EFOR, which is a function of the network position of an agent, can leverage more on the production efforts of the influencees.

Following Definition~\ref{def:EUF}, we define the {\em effort update function} (EUF) $F : [0,1]^n \to [0,1]^n$ for the general setting as, $F(\mathbf{x}) = T \circ \ell \circ g ( \mathbf{x} )$, where $T \circ \ell$ operates on the vector function $g$ element-wise.
Therefore, the question of existence of a Nash equilibrium of this effort game is the same as asking the question if the following fixed point equation has a solution: $ \mathbf{x} = F(\mathbf{x})$.

In the following, we provide a sufficient condition for existence of the Nash equilibrium, and its uniqueness.

\begin{lemma}
 \label{lma:F-continuous}
  For $p_i > 0$, for all $i \in N$, and continuously differentiable, $F(\mathbf{x})$ is continuous. 
\end{lemma}

\begin{proof}
 Given $p_i > 0$, for all $i \in N$, and is continuously differentiable. Therefore, the function $g$, defined in Equation (\ref{eq:g-general}), is continuous in $\mathbf{x}$. Using Facts~\ref{fact:invertible} and \ref{fact:T-continuous}, we see that the functions $\ell$ and $T$ are continuous. Hence, $F \equiv T \circ \ell \circ g$ is continuous in $\mathbf{x}$.
\end{proof}

\begin{theorem}[Sufficient Condition for a Nash Equilibrium]
 \label{thm:general-sufficient-one-NE}
  For $p_i > 0$, for all $i \in N$, and continuously differentiable, the effort game has at least one Nash equilibrium.
\end{theorem}

\begin{proof}
 From \Cref{lem:general-necessary-NE}, we see that the Nash equilibrium of the effort game is same as the fixed point of the equation, $\mathbf{x} = F(\mathbf{x})$.
 Since $F$ is continuous (Lemma~\ref{lma:F-continuous}), Brouwer's fixed point theorem immediately ensures a fixed point of the above equation to exist. Hence, the effort game has at least one Nash equilibrium.
\end{proof}

Let us use the shorthand $G \equiv \ell \circ g$. The following theorem provides a sufficient condition for the uniqueness of the Nash equilibrium.

\begin{theorem}[Sufficient Condition for unique Nash Equilibrium]
 \label{thm:general-linear-g-sufficient-unique-NE}
 If $\sup_{\mathbf{x}_0} |\nabla G (\mathbf{x}_0)| < 1$, then the Nash equilibrium effort profile $(x_i^*, x_{-i}^*)$ is unique and is given by Equation (\ref{eq:general-NE}).
\end{theorem}

\begin{proof}
 The key here is to show that $F$ is a contraction. We follow the steps of Theorem~\ref{thm:exponential-sufficiency} as follows:
 \begin{align*}
  ||F(\mathbf{x}) - F(\mathbf{y})|| \leqslant ||G(\mathbf{x}) - G(\mathbf{y})|| &\leqslant |\nabla G(\mathbf{x}_0)| \cdot ||\mathbf{x} - \mathbf{y}||.
 \end{align*}
 This is a contraction as $\sup_{\mathbf{x}_0} |\nabla G(\mathbf{x}_0)| < 1$.
\end{proof}

\subsubsection{Interpretation of the sufficient condition of the uniqueness}

The sufficient condition given by Theorem~\ref{thm:general-linear-g-sufficient-unique-NE} is a technical one. We now discuss an interesting geometric interpretation of this condition. By the Taylor expansion of $G$ with first order remainder term, we get,
\[G(\mathbf{x}) - G(\mathbf{y}) = \nabla G(\mathbf{x}_0) \cdot (\mathbf{x} - \mathbf{y}) .\]
Where $\mathbf{x}_0$ lies on the line joining $\mathbf{x}$ and $\mathbf{y}$. Using singular value decomposition, we get, $\nabla G(\mathbf{x}_0) = U_0 \Sigma_0 V_0^{\top}$. Therefore, for each pair of points $\mathbf{x}$ and $\mathbf{y}$, we can transform the space of efforts with a pure rotation as follows.
\begin{align*}
 G(\mathbf{x}) - G(\mathbf{y}) &= U_0 \Sigma_0 V_0^{\top} \cdot (\mathbf{x} - \mathbf{y}), \\
 \Rightarrow \quad U_0^{\top} (G(\mathbf{x}) - G(\mathbf{y})) &= \Sigma_0 \cdot (\mathbf{\bar{x}} - \mathbf{\bar{y}}), \mbox{ where } \mathbf{\bar{x}} = V_0^{\top} \mathbf{x}, \mathbf{\bar{y}} = V_0^{\top} \mathbf{y}\\
 \Rightarrow \quad R(\mathbf{\bar{x}}) - R(\mathbf{\bar{y}}) &= \Sigma_0 \cdot (\mathbf{\bar{x}} - \mathbf{\bar{y}}), \mbox{ where } R \equiv U_0^{\top} G V_0.
\end{align*}
Hence, for any pair of points $\mathbf{x}$ and $\mathbf{y}$, we can rotate the space so that the effect of the deviation to $\mathbf{x}$ from $\mathbf{y}$ can be captured by a weight on each of the coordinates in the rotated space. Here, the diagonal matrix $\Sigma_0$ contains the weights along its diagonal. 

Theorem~\ref{thm:general-linear-g-sufficient-unique-NE} says that for any point $\mathbf{x}_0$, if the absolute value of all the elements of this diagonal matrix is smaller than unity, the uniqueness of the Nash equilibrium is guaranteed. Let us denote the rotated vector of $\mathbf{x}_0$ by $\mathbf{z}_0 := V_0^{\top} \mathbf{x}_0$. The diagonal elements can be written w.r.t.\ the vectors in the rotated space as,
\[(\Sigma_0)_{ii} = \sum_{j \in E_i} h_{ij} \left . \left ( - \frac{1}{p_i} \frac{\partial^2 p_j}{\partial z_i^2} z_j \right ) \right |_{\mathbf{z}_0} = \frac{\partial \text{EFOR}(\mathbf{z}_0)}{\partial \mathbf{z}_{i,0}} . \]

In other words, the diagonal elements are the rate of change of EFOR at $\mathbf{z}_0$.
Having the rate of change of EFOR bounded by 1 is a sufficient condition for a unique Nash equilibrium. One can think of the EFOR as the product of two effects: (1) the rate of change in productivity, which increases the payoff of the influencees, (2) the reward share $h_{ij}$'s. The sufficient condition essentially says that the net effect should not be too large in order to guarantee unique equilibrium. 
%
\begin{wrapfigure}{R}{0.5\columnwidth}
\centering
 \includegraphics[width=0.4\columnwidth]{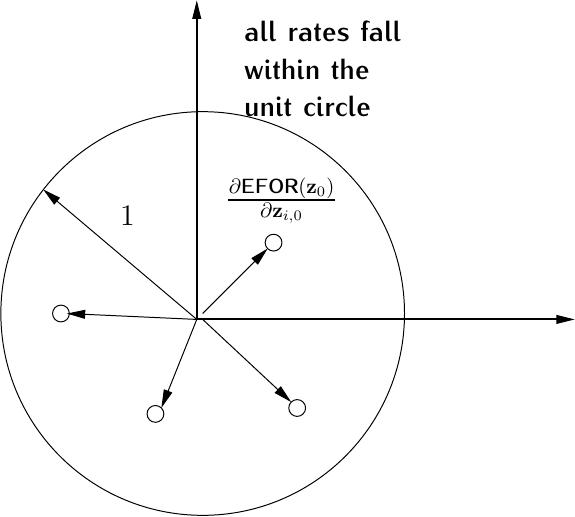}
\caption{Geometric interpretation of the condition for uniqueness of the Nash equilibrium.}
\label{fig:rou-NE}
\end{wrapfigure}
%

Figure~\ref{fig:rou-NE} shows a graphical illustration of the
phenomenon in polar co-ordinates, where the directions represent that of the vectors. The results say that if for any vector, the singular values of the Jacobian matrix of $G$ at that point lies entirely within the unit ball, then there exists an unique Nash equilibrium.
This is similar to the feedback loop gain of a feedback controller, where the closed loop gain being smaller than unity ensures stability. We find this natural parallel between notions of stability (from control theory) and uniqueness of Nash equilibria interesting.

\section{Related Work}

%
The study of effort levels in network games, where an agent's utility depends on actions of neighboring agents has recently received much attention \citep{Galeotti2010}. For example, \citet{ballester06} show how the level of activity of a given agent depends on the Bonacich centrality of the agent in the network, for a specific utility structure that results in a concave game. Our model differs in two aspects: (a) we have multiple types of efforts (namely production and communication) which has different nonlinear correlation among the agents, and (b) utilities are non-concave. In addition, our results give a design recipe for the reward sharing scheme. \citet{rogers2008strategic} analyzes the efficiency of equilibria in two specific types of games (i) `giving' and (ii) `taking', where an edge means utility is sent on an edge. A strategic model of effort is discussed in the public goods model of \citet{Bramoulle2007}, where utility is concave in individual agents' efforts, and the structures of the Nash and 
stable equilibria are shown. Their model applies to a very specific utility structure where the same benefit of the `public good' is experienced by all the first level neighbors on a graph. In our model, the individual utilities can be asymmetric, and depend on the efforts and 
reward shares in multiple levels on the graph.
Building on these efforts our utility model cleanly separate the effects of two types of influence, that we termed information and incentives.

The DARPA Red Balloon Challenge, and particularly the hierarchical network and specific reward structure used by the winning MIT team \citep{pickard-etal11MIT}, has led to a renewed interest in the analysis of effort exerted by agents in networks. The winning
team's strategy, utilized a recursive incentive mechanism.
Our results show that, in this case for example, too much reward sharing encourages managers to spend more time recruiting or managing and not enough time searching or working, though we do not study network formation games here.

The literature on strategic social network formation games and organizational design is vast \citep{jackson2009networks,harris2002organization}.
%
We use the Price of Anarchy (PoA) introduced by \citet{koutsoupias1999worst} to measure the sub-optimality in outcome efforts, as a function of network structure and incentives, due to the self interested nature of agents.
In the network contribution games literature, the PoA has been investigated in different contexts. \citet{Anshelevich2012} consider a model where an agent's contribution locally benefit the nodes who share an edge with him, and give existence and PoA results for pairwise equilibrium for different contribution functions. The PoA in cooperative network formation is considered by \citet{demaine2009price}, while \citet{roughgarden2005selfish, garg2005price} have considered the question in a selfish network routing context. Our setting is different from all of these since in our model the strategies are the efforts of the agents, which distinguishes it from the network formation and selfish routing literature, and we use multiple levels of information and reward sharing and study utilities that are asymmetric even for the neighboring nodes in the network, which distinguishes itself from the network contribution games.


\section{Summary and Future Work}

In this paper, we build on the papers by \citet{Bramoulle2007,ballester06} and develop an understanding of the effort levels in influencer-influencee networks. Taking a game theoretic perspective, we introduce a general utility model which results in a non-concave game, but are able to show results on the existence and uniqueness of Nash equilibrium efforts. For the ease of exposition, we focused on hierarchical networks, and with the EP model we found closed form expressions and bounds on the PoA for balanced hierarchies. These results give us the insight on the importance of communication in hierarchies on the design of efficient networks. At the same time, for a given network structure and communication level, we give a design recipe for the reward sharing in order to achieve highly productive output, and thereby minimize the PoA. 

The connection between matrix stability and uniqueness of Nash equilibria that arose in our work, is of particular interest to us for future research. In particular, for the general networks there was a direct interpretation of the uniqueness condition in terms of a Jacobian matrix stability. This stability property is directly related to the contraction property that shows that agents following local updates on effort levels will converge to the Nash equilibrium another desirable property. Pursuing these connections in the investigation of reward share design where individual employees behave in a strategic way in organizational networks is an important direction of future research.




\pagebreak

\appendix

\section*{Appendices}
\setcounter{section}{0}

\section{Proofs for the Exponential Productivity Model}
\label{app:EP-proofs}

\subsection{Proof of Theorem \ref{thm:exp-NE-necessary}}

\begin{proof}
 The argument for the existence of a Nash equilibrium is straightforward in this particular setting. We see that because of the hierarchical structure of the network, the leaf nodes will always put unit effort, i.e., $\mathbf{x}^*_{\text{leaves}} = 1$. To compute the equilibrium in the level above the leaves one can run a backward induction algorithm to maximize \ref{eq:general-payoff-first} at each level, where the equilibrium efforts in the levels below is already computed by the algorithm. Since, all $p_i$'s are bounded and the maximization is over $x_i \in [0,1]$, a compact space, maxima always exists. Hence, a Nash equilibrium always exists.

 Now we show that a Nash equilibrium profile $(x_i^*, x_{-i}^*)$ must satisfy Equation (\ref{eq:exponential-NE}). For notational convenience, we drop the arguments of $p_i$ and $p_{ij}$, which are functions of $\mathbf{x}_{P_{\theta \to i}}$ and $\mathbf{x}_{P_{i_- \to j}}$ respectively.
 Each agent $i \in N$ solves the following optimization problem.
  \begin{align}
\label{eq:opt-necessary}
\begin{array}{cc}
 \max_{x_i} & u_i(x_i, x_{-i}) \\
 \mbox{ s.t. } & x_i \geqslant 0
\end{array}
  \end{align}
  Combining Equations (\ref{eq:general-payoff-first}), (\ref{eq:f-in-EP}), and (\ref{eq:exponential-productivity}), we get,
  \[u_i(x_i,x_{-i}) = p_i(\mathbf{x}_{P_{\theta \to i}}) \left (x_i - \frac{x_i^2}{2} - b \frac{(1-x_i)^2}{2} \right ) + \sum_{j \in T_i \setminus \{i\}} h_{ij} p_j(\mathbf{x}_{P_{\theta \to j}}) x_j . \]
 Note that we have relaxed the constraint from $0 \leqslant x_i \leqslant 1$. The first additive term in the utility function has the peak at $x_i = 1$. The second term has $e^{\beta x_i}$ in the $p_j$, which is decreasing in $x_i$. Therefore, the optimal $x_i$ that maximizes this utility will be $\leqslant 1$. Hence, in this problem setting, the optimal solution for both the exact and the relaxed problems is the same. So, it is enough to consider the above problem. For this non-linear optimization problem, we can write down the Lagrangian as follows.
 \[{\cal L} = u_i(x_i, x_{-i}) + \lambda_i x_i, \ \lambda_i \geqslant 0.\]
 The KKT conditions for this optimization problem (\ref{eq:opt-necessary}) are:
 \begin{align}
  \frac{\partial {\cal L}}{\partial x_i} &= 0, \Rightarrow \frac{\partial }{\partial x_i} u_i(x_i, x_{-i}) + \lambda_i = 0,  & \label{eq:KKT-1} \\
  \lambda_i x_i &= 0, & \mbox{ complementary slackness.}\label{eq:KKT-2} 
 \end{align}
 \noindent {\em Case 1:} $\lambda_i = 0$, then from Equation (\ref{eq:KKT-1}) we get, 
\begin{align}
  & \quad p_i (1 - x_i + b (1-x_i)) + \sum_{j \in T_i \setminus \{i\}} h_{ij} \frac{\partial p_j}{\partial x_i} x_j = 0 \nonumber \\
 \Rightarrow & \quad p_i (1+b) (1-x_i) - \beta \sum_{j \in T_i \setminus \{i\}} h_{ij} p_j x_j = 0 \nonumber \\
 \Rightarrow & \quad 1-x_i = \frac{\beta}{1+b} \sum_{j \in T_i \setminus \{i\}} h_{ij} p_{ij} x_j,  \mbox{ with $p_{ij}$ as defined} \nonumber \\
 \Rightarrow & \quad x_i = 1 - \frac{\beta}{1+b} \sum_{j \in T_i \setminus \{i\}} h_{ij} p_{ij} x_j. \label{eq:cond-1}
\end{align}
  \noindent {\em Case 2:} $\lambda_i > 0$, then from Equation (\ref{eq:KKT-2}) we get $x_i = 0$, and from Equation (\ref{eq:KKT-1}), $\frac{\partial }{\partial x_i} u_i(x_i, x_{-i}) < 0$.
 Carrying out the differentiation as in Equation (\ref{eq:cond-1}) we get,
 \begin{align}
  0 = x_i &> 1 - \frac{\beta}{1+b} \sum_{j \in T_i \setminus \{i\}} h_{ij} p_{ij} x_j. \label{eq:cond-2}
 \end{align}
 \[\therefore \ x_i = \left [ 1 - \frac{\beta}{1+b} \sum_{j \in T_i \setminus \{i\}} h_{ij} p_{ij} x_j \right ]^+ .\]
 Since this condition has to hold for all nodes $i \in N$, the equilibrium profile $(x^*_i, x^*_{-i})$ must satisfy the above equality. 
\end{proof}

\subsection{Proof of Theorem~\ref{thm:exponential-sufficiency}}

We prove this theorem via the following Lemma.
\begin{lemma}
  \label{lem:contraction}
 If $\beta < \sqrt{\frac{1+b}{h_{\max}(T)}}$, the function $F$ is a contraction.
\end{lemma}

\begin{proof}
 The Taylor series expansion of $g$ with a first order remainder term is as follows. There exists a point $\mathbf{x}_0$ that lies on the line joining $\mathbf{x}$ and $\mathbf{y}$, such that,
 \[g(\mathbf{x}) = g(\mathbf{y}) + \nabla g(\mathbf{x}_0) \cdot (\mathbf{x} - \mathbf{y}).\]
 Where, $\nabla g(\mathbf{x}_0)$ is the Jacobian matrix.
 \[
\nabla g(\mathbf{x}_0) = \left. \left( 
\begin{array}{ccc}
 \frac{\partial g_1}{\partial x_1} & \dots & \frac{\partial g_1}{\partial x_n} \\
 \vdots & \ddots & \\
 \frac{\partial g_n}{\partial x_1} & \dots & \frac{\partial g_n}{\partial x_n}
\end{array}
\right) \right|_{\mathbf{x}_0}
\]
 In order to show that $F$ is a contraction, we note that $F$ is a truncation of $g$. Hence, $||F(\mathbf{x}) - F(\mathbf{y})|| \leqslant ||g(\mathbf{x}) - g(\mathbf{y})||$, for all $\mathbf{x}, \mathbf{y} \in [0,1]^n$. Let us consider the following term,
 \begin{align}
   ||F(\mathbf{x}) - F(\mathbf{y})|| & \leqslant ||g(\mathbf{x}) - g(\mathbf{y})|| \leqslant |\nabla g(\mathbf{x}_0)| \cdot ||\mathbf{x} - \mathbf{y}|| \label{eq:contraction}
 \end{align}
 Where the matrix norm $|\nabla g(\mathbf{x}_0)|$ is the largest singular value of the Jacobian matrix $\nabla g(\mathbf{x}_0)$. We see that in our special structure in the problem, this matrix is upper triangular, hence the diagonal elements are the singular values. Suppose, the $k$-th diagonal element yields the largest singular value.
  \begin{align*}
   |\nabla g(\mathbf{x}_0)| &= \left. \frac{\partial g_k}{\partial x_k} \right|_{\mathbf{x}_0} = \left. \frac{\beta^2}{1+b} \sum_{j \in T_k \setminus \{k\}} h_{kj} p_{kj} x_j \right|_{\mathbf{x}_0} \\
  \Rightarrow \quad \sup_{x_0} |\nabla g(\mathbf{x}_0)| &\leqslant \frac{\beta^2}{1+b} \cdot h_{\max}(T) < 1, \quad \mbox{ since } \beta^2 < \frac{1+b}{h_{\max}(T)}.
  \end{align*}
 The first inequality above holds due to the fact that $p_{kj}$'s and $x_j$'s are $\leqslant 1$, and by the definition of $h_{\max}(T)$.
 Hence, from Equation (\ref{eq:contraction}), we get that $F$ is a contraction.
\end{proof}

\begin{proof}[of Theorem~\ref{thm:exponential-sufficiency}]
 We know from \Cref{thm:exp-NE-necessary} that a Nash equilibrium exists. Under the sufficient condition given by Lemma~\ref{lem:contraction}, the fixed point of $\mathbf{x} = F(\mathbf{x})$ is unique. Therefore, the Nash equilibrium is also unique, and is given by Equation (\ref{eq:exponential-NE}).
\end{proof}

\subsection{Proof of \Cref{cor:node-NE-complexity}}
\begin{proof}
 To compute the equilibrium production effort $x_i^*$, node $i$ needs to compute \Cref{eq:exponential-NE}. This requires to compute the equilibrium efforts for each node in his subtree $T_i$. Because of the fact that $x_i^*$ depends only on the equilibrium efforts of the subtree below $i$, we can apply the backward induction method starting from the leaves towards the root of this sub-hierarchy $T_i$. The worst-case complexity of such a backward induction occurs when the sub-hierarchy is a line. In such a case the complexity would be $|T_i| (|T_i| - 1) / 2 = O(|T_i|^2)$ --- to compute $x_i^*$ we need the equilibrium effort of every node below $i$ in the hierarchy, and each such node needs computation equal to its distance from the leaf of this line. 
 In order to compute the equilibrium efforts of the whole network, it is enough 
to determine the equilibrium effort at the root
because this would, in the process, determine
the equilibrium efforts of each node in the hierarchy. 
The worst-case complexity of finding the equilibrium effort at the root is $O(n^2)$ and therefore the worst-case complexity of computing the equilibrium efforts of the whole network is also $O(n^2)$.
\end{proof}

%
%



\section{Proofs of the price of anarchy results in balanced hierarchies}
\label{app:B}

\subsection{Proof of Theorem~\ref{thm:poa}}
We prove this theorem via the following lemma, which finds out the optimal effort profile for $\beta$ above a threshold.

\begin{lemma}[Optimal Efforts]
 \label{lem:optimal-effort}
 For a balanced $d$-ary hierarchy with depth $D$, any optimal effort profile has $x_{D+1}^{\text{OPT}} = 1$. When $-\ln \left( 1 - \frac{1}{d} \right) \leqslant \beta < \infty$, the optimal effort profile is $x_i^{\text{OPT}} = 0,\ \forall i = 1, \dots, D$, and $x_{D+1}^{\text{OPT}} = 1$. 
\end{lemma}

\begin{proof}
 The social outcome for a given effort vector $\mathbf{x}$ on the balanced hierarchy is as follows. Since, the hierarchy is understood here, we use $SO(\mathbf{x})$ instead of $SO(\mathbf{x}, \texttt{BALANCED})$.
 \begin{align*}
  SO(\mathbf{x}) &= x_1 + d e^{-\beta x_1} x_2 + d^2 e^{-\beta (x_1 + x_2)} x_3 + \dots + d^D e^{- \beta (\sum_{i=1}^D x_i)} x_{D+1}.
 \end{align*}
 It is clear that for any effort profile of the other nodes the effort at the leaves that maximizes the above expression is $x_{D+1} = 1$. This proves the first part of the lemma. Hence we can simplify the above expression by,
 \begin{align}
  SO(\mathbf{x}) &= x_1 + d e^{-\beta x_1} x_2 + d^2 e^{-\beta (x_1 + x_2)} x_3 + \dots + d^D e^{-\beta (\sum_{i=1}^D x_i)} \nonumber \\
    &= x_1 + d e^{-\beta x_1} x_2 + \dots + d^{D-1} e^{-\beta (\sum_{i=1}^{D-1} x_i)} (x_D + d e^{-\beta x_D}) \label{eq:soc-output} \\
    &\leqslant x_1 + d e^{-\beta x_1} x_2 + \dots + d^{D-1} e^{-\beta (\sum_{i=1}^{D-1} x_i)} \cdot d . \nonumber
 \end{align}  
 The last inequality occurs since $\beta \geqslant -\ln \left( 1 - \frac{1}{d}\right)$, and $x_D = 0$ meets this inequality with a equality. Also since $\beta \geqslant -\ln \left( 1 - \frac{1}{d}\right)$ implies that $\beta \geqslant -\ln \left( 1 - \frac{1}{d^k}\right)$, for all $k \geqslant 2$, the next inequality will also be met by $x_{D-1} = 0$ as shown below.
 \begin{align*}
  SO(\mathbf{x}) &= x_1 + d e^{-\beta x_1} x_2 + \dots + d^{D-1} e^{-\beta (\sum_{i=1}^{D-1} x_i)} \cdot d \\
    &= x_1 + d e^{-\beta x_1} x_2 + \dots + d^{D-2} e^{-\beta (\sum_{i=1}^{D-2} x_i)} (x_{D-1} + d^2 e^{-\beta x_{D-1}}) \\
    &\leqslant x_1 + d e^{-\beta x_1} x_2 + \dots + d^{D-2} e^{-\beta (\sum_{i=1}^{D-2} x_i)} \cdot d^2 .
 \end{align*} 
 This inequality is also achieved by $x_{D-1} = 0$. We can keep on reducing the terms from the right in the RHS of the above equation, and in all the reduced forms, $x_i = 0, i = D-1, D-2, \dots, 1$ will maximize the social output expression. Hence proved.
\end{proof}

\begin{proof}[of Theorem~\ref{thm:poa}]
 {\em Case 1 ($0 \leqslant \beta \leqslant 1$):} From Lemma~\ref{lem:optimal-effort}, $x_{D+1} = 1$ for optimal effort. However, for any equilibrium effort profile $x_{D+1} = 1$ as well. Therefore we consider the equilibrium effort of the nodes at level $D$.
  \begin{equation}
    x_{D} = 1 - \frac{\beta}{1+b} d e^{-\beta x_D} h_{D,D+1}.   \label{eq:poa-equilibrium}
  \end{equation}
  The constraint for unique equilibrium demands that $d h_{D,D+1} \leqslant (1+b)/\beta^2$, which makes $\frac{\beta}{1+b} d h_{D,D+1} \leqslant 1/\beta$, while $1/\beta \geqslant 1$. So, we have the liberty of choosing the right $h_{D,D+1}$ to achieve any $x_D \in [0,1]$, and in particular, the $x_D^{\text{OPT}}$. We apply backward induction on the next level above.
  \[x_{D-1} = 1 - \frac{\beta}{1+b} [d e^{-\beta x_{D-1}} x_{D} h_{D-1,D} + d^2 e^{-\beta (x_{D-1} + x_{D})} h_{D-1,D+1} ]. \]
  The constraints are $d h_{D-1,D} + d^2 h_{D-1,D+1} \leqslant (1+b)/\beta^2$. We claim that any $x_{D-1} \in [0,1]$ is achievable here as well. To show that, put $h_{D-1,D} = 0$. The above equation becomes then,
  \begin{align*}
   x_{D-1} &= 1 - \frac{\beta}{1+b} d^2 e^{-\beta (x_{D-1} + x_{D})} h_{D-1,D+1} \\
	  &= 1 - \frac{\beta}{1+b} d^2 e^{-\beta x_{D-1}} \frac{1+b}{d \beta h_{D,D+1}} h_{D-1,D+1}, \mbox{ from the earlier expression} \\
	  &= 1 - \frac{d h_{D-1,D+1}}{h_{D,D+1}} e^{-\beta x_{D-1}}
  \end{align*}
  This again can satisfy any $x_{D-1}$, since the coefficient of the exponential term can be made anywhere between 0 and 1. It can be made 0 by choosing $h_{D-1,D+1} = 0$, and 1 by choosing $\frac{d h_{D-1,D+1}}{h_{D,D+1}} = 1$ which is feasible, since $d^2 h_{D-1,D+1} = d h_{D,D+1} \leqslant (1+b)/\beta^2$.
  
  In the similar way we can continue the induction till the root and can make $\mathbf{x}^* = \mathbf{x}^{\text{OPT}}$. Hence, PoA $=1$.

 {\em Case 2 ($1 < \beta < \infty$):} We note that this region of $\beta$ falls in the region specified by Lemma~\ref{lem:optimal-effort}. Hence the optimal effort is 1 for all the leaves and 0 for everyone else. Hence, the optimal social output is given by $d^D$. The equilibrium effort for the leaves, $x_{D+1} = 1$. However, Equation (\ref{eq:poa-equilibrium}) may not be satisfiable for any $x_D$ since $1/\beta < 1$. In order to push the solution as close to zero as possible, we choose $h_{D,D+1} = (1+b)/\beta^2$, and plug it in Equation (\ref{eq:poa-equilibrium}), and the solution is given by $\xi(\beta)$ (recall Equation (\ref{eq:xi-defn})) and the solution set is singleton under this condition. The social output is $d^D$, which is the numerator of the PoA expression. The denominator is given by the social output at the Nash equilibrium, which we will try to lower bound. From Equation (\ref{eq:soc-output}), for the equilibrium, we know that $x_D = \xi(\beta)$. Therefore, 
  \[x_D + d e^{-\beta x_D} = x_D + d \beta (1-x_D) = d \beta + (1 - d \beta) \xi(\beta).\]
 At the same time, we see that the leftmost expression is convex in $x_D$, which can be lower bounded by the minima, given by,
 \[ x_D + d e^{-\beta x_D} \geqslant \frac{1}{\beta} (1+\ln(d \beta)).\]
  Combining the two, a tight lower bound of the expression would be,
  \[x_D + d e^{-\beta x_D} \geqslant \max \left \{ \frac{1}{\beta} (1 + \ln (d \beta)), d \beta +  (1-d \beta) \xi(\beta) \right \} = \phi(d,\beta) .\]	
  Plugging this lower bound in Equation (\ref{eq:soc-output}), we see that,
  \begin{align*}
  \lefteqn{SO(\mathbf{x})} \\ 
    &\geqslant x_1 + d e^{-\beta x_1} x_2 + \dots + d^{D-1} e^{-\beta (\sum_{i=1}^{D-1} x_i)} \cdot \phi(d,\beta) \\
    &= x_1 + d e^{-\beta x_1} x_2 + \dots + d^{D-2} e^{-\beta (\sum_{i=1}^{D-2} x_i)} \cdot (x_{D-1} + d \phi(d,\beta) e^{-\beta x_{D-1}} )
 \end{align*} 

 Let us consider the last term within parenthesis.
 \begin{align*}
  & x_{D-1} + d \phi(d,\beta) e^{-\beta x_{D-1}} \\
  &= x_{D-1} + d \phi(d,\beta) \frac{\beta}{\xi(\beta)} (1 - x_{D-1}) \\
  &\geqslant x_{D-1} + d \phi(d,\beta) \beta (1 - x_{D-1}) , \ \mbox{ as } \xi(\beta) \leqslant 1 \\
  &= d \phi(d,\beta) \beta + (1 - d \phi(d,\beta) \beta) x_{D-1} \\
  &\geqslant d \phi(d,\beta) \beta + (1 - d \phi(d,\beta) \beta) \xi(\beta) 
 \end{align*}

 The first equality comes since we can make the equilibrium $x_{D-1}$ s.t., $x_{D-1}  =  1 - \frac{\xi(\beta)}{\beta} e^{-\beta x_{D-1}}$, by choosing $d h_{D-1,D} = (1+b)/\beta^2, d^2 h_{D-1,D+1} = 0$. Also, since $\xi(\beta) \leqslant 1$, we conclude, $x_{D-1} \geqslant x_D = \xi(\beta)$, which gives the second inequality above.
 On the other hand, using the fact that the expression $x_{D-1} + d \phi(d,\beta) e^{-\beta x_{D-1}}$ is convex in $x_{D-1}$, it can be lower bounded by, $\frac{1}{\beta} (1+\ln (d \phi(d,\beta) \beta))$. Combining this and the above inequality, we get the following.
   \begin{align*}
  SO(\mathbf{x}) &\geqslant x_1 + d e^{-\beta x_1} x_2 + \dots + d^{D-2} e^{-\beta (\sum_{i=1}^{D-2} x_i)} \cdot \phi(d \cdot \phi(d,\beta),\beta) \\
    & \qquad \vdots \quad \mbox{ repeating the steps above} \\
    &\geqslant t_D(d,\beta), \qquad \mbox{as defined in Equation (\ref{eq:t_D}).}
 \end{align*} 
 Therefore the PoA $\leqslant \frac{d^D}{t_D(d,\beta)}$.
\end{proof}

\section{Proofs for general networks}
\subsection{Proof of Lemma \ref{lem:general-necessary-NE}}

\begin{proof}
 We follow the line of proof of Theorem~\ref{thm:exp-NE-necessary}. 
 Each agent $i \in N$ is solving the following optimization problem.
  \begin{align}
\label{eq:opt-necessary-genl}
\begin{array}{cc}
 \max_{x_i} & u_i(x_i, x_{-i}) \\
 \mbox{ s.t. } & 0 \leqslant x_i \leqslant 1
\end{array}
  \end{align}
 This is a non-linear optimization problem. Hence we can write down the Lagrangian as follows.
 \[{\cal L} = u_i(x_i, x_{-i}) + \lambda_i x_i + \gamma_i (1-x_i), \ \lambda_i, \gamma_i \geqslant 0.\]
 The KKT conditions are necessary for this optimization problem (\ref{eq:opt-necessary-genl}), which are the following.
 \begin{align}
  \frac{\partial {\cal L}}{\partial x_i} &= 0, \nonumber \\
  \Rightarrow \frac{\partial }{\partial x_i} u_i(x_i, x_{-i}) + \lambda_i - \gamma_i &= 0,  \label{eq:KKT-1-genl} \\
  \lambda_i x_i = 0, \quad \gamma_i (1-x_i) &= 0. \label{eq:KKT-2-genl}
 \end{align}
 \noindent {\em Case 1:} $\lambda_i = 0, \gamma_i = 0$, then from Equation (\ref{eq:KKT-1-genl}) we get, 
\begin{align}
 & \quad \frac{\partial }{\partial x_i} u_i(x_i, x_{-i}) = 0 \nonumber \\
 \Rightarrow & \quad p_i f'(x_i) + \sum_{j \in E_i} h_{ij} \frac{\partial p_j}{\partial x_i} x_j = 0 \nonumber \\
 \Rightarrow & \quad f'(x_i) = \sum_{j \in E_i} h_{ij} \left ( - \frac{1}{p_i} \frac{\partial p_j}{\partial x_i} x_j \right ) = g_i(\mathbf{x}) \nonumber \\
 \Rightarrow & \quad x_i = \ell \circ g_i(\mathbf{x}), \mbox{ from the definition of } \ell \label{eq:cond-1-genl}
\end{align}
  \noindent {\em Case 2:} $\lambda_i > 0, \gamma_i = 0$, then from Equation (\ref{eq:KKT-2-genl}) we get $x_i = 0$, and from Equation (\ref{eq:KKT-1-genl}),
  \[\frac{\partial }{\partial x_i} u_i(x_i, x_{-i}) < 0.\]
 Carrying out the differentiation as in Equation (\ref{eq:cond-1}), we get,
 \begin{align}
  f'(x_i) < g_i(\mathbf{x}) &\Rightarrow 0 = x_i > \ell \circ g_i(\mathbf{x}), \mbox{ since $f$ is concave} \nonumber \\
  \Rightarrow x_i &= T \circ \ell \circ g_i(\mathbf{x}), \mbox{ where $T$ is the truncation function.} \label{eq:cond-2-genl}
 \end{align}
  \noindent {\em Case 3:} $\lambda_i = 0, \gamma_i > 0$, then from Equation (\ref{eq:KKT-2-genl}) we get $x_i = 1$, and from Equation (\ref{eq:KKT-1-genl}),
  \[\frac{\partial }{\partial x_i} u_i(x_i, x_{-i}) > 0.\]
 Carrying out similar steps as before, we get,
 \begin{align}
  f'(x_i) > g_i(\mathbf{x}) &\Rightarrow 1 = x_i < \ell \circ g_i(\mathbf{x}) \nonumber \\
  \Rightarrow x_i &= T \circ \ell \circ g_i(\mathbf{x}), \mbox{ where $T$ is the truncation function.} \label{eq:cond-3-genl}
 \end{align}
 \noindent {\em Case 4:} $\lambda_i > 0, \gamma_i > 0$, this cannot happen since it will lead to a contradiction $0 = x_i = 1$.
 Therefore, combining Equations (\ref{eq:cond-1-genl}), (\ref{eq:cond-2-genl}), and (\ref{eq:cond-3-genl}), we get,
 \[x_i^* = T \circ \ell \circ g_i(\mathbf{x}^*), \ \forall i \in N.\]
 Hence proved.
\end{proof}

\end{document}